\documentclass[12pt,a4paper]{article}
\usepackage{graphicx}
\usepackage{amsmath}
\usepackage{amssymb}
\usepackage{enumerate}
\usepackage{cspsymb}
\usepackage{stackrel}
\usepackage[T1]{fontenc}
\usepackage{tikz}
\usepackage[margin=1.2in]{geometry}
\usepackage{hyperref}
\usetikzlibrary{automata,positioning}
\usepackage[appendix=inline]{apxproof}

\newcommand{\XX}{\mathcal{X}}

\newcommand{\LL}{\mathcal{L}}

\newcommand{\FF}{\mathcal{F}}

\newcommand{\A}{\mathcal{A}}

\begin{document}

\title{Widths of regular and context-free languages}
\author{David Mestel \\ University of Luxembourg \\ david.mestel@uni.lu}
\date{}
\maketitle              
\begin{abstract}
Given a partially-ordered finite alphabet $\Sigma$ and a language $L\subseteq
\Sigma^*$, how large can an antichain in $L$ be (where $L$ is given the
lexicographic ordering)? More precisely, since $L$ will in general be infinite,
we should ask about the rate of growth of maximum antichains consisting of words
of length $n$. This fundamental property of partial orders is known as the
width, and in a companion work \cite{mestel2018quantifying} we show that the
problem of computing the information leakage permitted by a deterministic
interactive system modeled as a finite-state transducer can be reduced to the
problem of computing the width of a certain regular language. In this paper, we
show that if $L$ is regular then there is a dichotomy between polynomial and
exponential antichain growth. We give a polynomial-time algorithm to distinguish
the two cases, and to compute the order of polynomial growth, with the language
specified as an NFA. For context-free languages we show that there is a similar
dichotomy, but now the problem of distinguishing the two cases is undecidable.
Finally, we generalise the lexicographic order to tree languages, and show that
for regular tree languages there is a trichotomy between polynomial, exponential
and doubly exponential antichain growth.
\end{abstract}
\newcommand{\MM}{\mathcal{M}}
\newcommand{\RR}{\mathcal{R}}
\newcommand{\Aa}{\mathcal{A}}
\newcommand{\Bb}{\mathcal{B}}
\newcommand{\Iln}{I_{\leq n}}
\newcommand{\UU}{\mathcal{U}}
\renewcommand{\XX}{\mathcal{X}}
\newcommand{\Ss}{\mathcal{S}}
\newcommand{\prodstar}{\stackrel{*}{\Rightarrow}}
\newcommand{\arel}{\stackrel[\Aa]{}{\rightarrow}}
\newcommand{\arels}{\stackrel[\Aa]{*}{\rightarrow}}
\newcommand{\sub}[2]{#1 \leftarrow #2}

\theoremstyle{plain}
\newtheoremrep{theorem}{Theorem}
\newtheoremrep{lemma}[theorem]{Lemma}
\newtheoremrep{definition}[theorem]{Definition}
\newtheoremrep{proposition}[theorem]{Proposition}
\newtheoremrep{corollary}[theorem]{Corollary}

\section{Introduction}

Computing the size of the largest antichain (set of mutually incomparable
elements) is the `central' problem in the extremal combinatorics of partially ordered sets
(posets)~\cite{West1982}.  In addition to some general
theory~\cite{KLEITMAN197147}, it has attracted study for a variety of specific
sets, beginning with Sperner's Theorem on subsets of $\{1,\ldots,n\}$ ordered by
inclusion~\cite{sperner1928satz,canfield1978problem,peck1979maximum}, and for
random posets~\cite{brightwell1992random}.  The size of the largest antichain in
a poset $L$ is called the \emph{width} of $L$.

In this work we study languages (regular or context-free) over finite partially
ordered alphabets, with the lexicographic partial order. Since such languages
will in general contain infinite antichains, we study the sets $L_{=n}$ of words
of length $n$, and ask how the width of $L_{=n}$ grows with $n$; we call this
the \emph{antichain growth} rate of $L$.

In addition to its theoretical interest, the motivation for this work is the
study of quantified information flow in computer security: we wish 
to know whether a pair of isolated agents interacting with a common central system 
(for example different programs 
running on a single computer and communicating with the operating system) can 
obtain any information about each other's actions, and if so how much.
In a companion work \cite{mestel2018quantifying} we show that
if the central system is modeled as a deterministic finite-state
transducer then this leakage is equivalent to the width of a certain regular language
(roughly speaking, antichains corresponding to consistent sets of observations for
one agent). The
dichotomy we obtain in this paper thus corresponds to a dichotomy between
logarithmic and linear information flow.

In Section \ref{sec:lang} we set out basic definitions and results on the
lexicographic order, antichains and antichain growth. In Section \ref{sec:reg}
we show that for regular languages there is a dichotomy between polynomial and
exponential antichain growth, and give a polynomial-time algorithm for
distinguishing the two cases. In Section \ref{sec:precisegrowth} we give a
polynomial-time algorithm to compute the order of polynomial antichain growth.
In Section \ref{sec:contextfree} we show that for context-free languages there
is a similar dichotomy between polynomial and exponential antichain growth, but
that the problem of distinguishing the two cases is undecidable. In
Section \ref{sec:tree} we show that for regular tree languages there is a
trichotomy between polynomial, exponential and doubly exponential antichain
growth.  Finally in Section \ref{sec:openprobs} we discuss open problems.

\section{Languages, lexicographic order and antichains}\label{sec:lang}

\begin{definition}\label{def:lex}Let $\Sigma$ be a finite alphabet equipped with a partial
order $\preceq$.  Then the lexicographic partial order induced by $\preceq$ on
$\Sigma^*$ is the relation $\preceq$ given by
\begin{enumerate}[(i)]\item$\epsilon\preceq w$ for all $w\in\Sigma^*$ (where
$\epsilon$ is the empty word), and
\item For any $x,y\in\Sigma, w,w'\in\Sigma^*$, we have $xw\preceq yw'$ if and
only if either $x \prec y$ or $x=y$ and $w\preceq w'$.
\end{enumerate}
\end{definition}

If words $x$ and $y$ are comparable in this partial order we write $x\sim y$.  If $x$ is a
prefix of $y$ we write $x\leq y$.

For a language $L$, we will often write $L_{=n}$ to denote the set $\{w\in L\mid
|w|=n\}$ (with corresponding definitions for $L_{<n}$, etc.), and $|L|_{=n}$ for
$\left|L_{=n}\right|$.

The main subject of this work is \emph{antichains}, that is sets of words which are
mutually incomparable.  It will sometimes be useful also to consider
\emph{quasiantichains},
which are sets of words which are incomparable except
that the set may include prefixes (note that this is not a standard term).
The opposite of an antichain is a \emph{chain}, in which all elements are 
comparable.

\begin{definition}A language $L$ is an \emph{antichain} if for every $l_1, l_2
\in L $ with $l_1\neq l_2$ we have $l_1 \not\sim l_2$.  $L$ is a
\emph{quasiantichain} if for every
$l_1,l_2\in L$ we have either $l_1\leq l_2$, $l_2\leq l_1$  or $l_1\not \sim
l_2$.  $L$ is a \emph{chain} if for all $l_1,l_2\in L$ we have $l_1\sim l_2$.
\end{definition}

It is easy to see that the property of being an antichain is preserved by the
operations of prefixing, postfixing and concatenation.

\begin{lemma}[(Prefixing)]Let $w,w_1,w_2$ be any words.  Then $w_1\sim w_2$ if and only if
$ww_1\sim ww_2$.  Hence for any language $L$, $wL$ is an antichain
(respectively quasiantichain) if and only if $L$ is
an antichain (quasiantichain).\end{lemma}

\begin{lemma}[(Postfixing)]Let $w,w_1,w_2$ be any words.  Then $w_1\sim w_2$ if
$w_1w\sim w_2w$.  Hence for any language $L$, $Lw$ is an antichain if $L$ is an
antichain.\end{lemma}

\begin{lemma}[(Concatenation)]Let $w_1,w_2,w_1',w_2'$ be any words such that
$w_1\not\leq w_2$ and $w_2\not \leq w_1$.  Then $w_1w_1'\sim w_2w_2'$ if and
only if $w_1\sim w_2$.  Hence if $L_1$ and $L_2$ are antichains then $L_1L_2$
is an antichain.\end{lemma}

Clearly the property of being an antichain is not preserved by Kleene star,
since $L^*$ will contain prefixes for any non-empty $L$.  The best we can hope
for is that $L^*$ is a quasiantichain.

\begin{lemmarep}[(Kleene star)]\label{lem:kleene}Let $L$ be an antichain.  Then $L^*$ is a quasiantichain.
\end{lemmarep}
\begin{proof}
Suppose $w_1\sim w_2$ with $w_1,w_2\in L^*, w_1\not\leq w_2$ and
$w_2\not\leq w_1$ with $|w_1+w_2|$ minimal.  Say $w_i=w_i'w_i''$ with $w_i'\in
L, w_i''\in L^*$.  By minimality we have $w_1'\neq w_2'$, and since $L$ is an
antichain we also have $w_1'\not\sim w_2'$.  Hence by the concatenation lemma
$w_1'w_1''\not\sim w_2'w_2''$, a contradiction.
\end{proof} 

Ultimately we are going to care about the size of antichains inside particular
languages.  Since these will often be unbounded, we choose to ask about the rate
of growth; that is, if $L_1,L_2,L_3,\ldots\subseteq L$ are antichains such that
$L_i$ consists of words of length $i$, how quickly can $|L_i|$ grow with $i$?
We will call $\bigcup_i L_i$ an \emph{antichain family} and ask whether it grows
exponentially, polynomially, etc.

\begin{definition}A language $L$ is an \emph{antichain family} if for each $n$
the set $L_{=n}$ of words in $L$ of length $n$ is an antichain.\end{definition}

\begin{definition}A language $L$ is \emph{exponential} (or \emph{has exponential
growth}) if there exists some
$\epsilon>0$ such that 
\[\limsup_{n\rightarrow\infty} \frac{|L|_{=n}}{2^{\epsilon n}} > 0,\]
and the supremum of the set of $\epsilon$ for which this holds is the \emph{order} 
of exponential growth.

$L$ is \emph{polynomial} (or \emph{has polynomial growth}) if there exists some $k$ such that 
\[\limsup_{n\rightarrow\infty} \frac{|L|_{=n}}{n^k} < \infty.\]
If $0 < \limsup_{n\rightarrow\infty} \frac{|L|_{=n}}{n^k} < \infty$ then we say
that $L$ has polynomial growth of order $k$.

For notational convenience, we will sometimes later adopt the convention that a
language $L$ which is finite (and so $\limsup_{n\rightarrow
\infty}\frac{|L|_{=n}}{n^k}=0$ for all $k$) has polynomial growth of order $-1$.
\end{definition}

A reasonable alternative choice of notation would have been to define the
quantity $w_n$ to be the size of the largest antichain consisting of words of
length $n$, and then ask about the growth of the series $w_1,w_2,\ldots$.  This
is clearly equivalent to the definitions we have given above.

Antichain growth generalises the classical notion of language growth, which 
is just antichain growth with respect to the discrete partial order 
(in which all elements of $\Sigma$ are incomparable).

Note that we will sometimes use other characterisations that are clearly
equivalent; for instance $L$ has exponential growth if and only if there is some
$\epsilon$ such that $|L|_{=n}>2^{\epsilon n}$ infinitely often.  We will
sometimes refer to a language which is not polynomial as `super-polynomial', or
as having `growth beyond all polynomial orders'.  Of course there exist
languages whose growth rates are neither polynomial nor exponential; for
instance $|L|_{=n}=\Theta(2^{\sqrt{n}})$.

\begin{definition}A language $L$ \emph{has exponential antichain growth} if
there is an exponential antichain family $L'\subseteq L$.  $L$ \emph{has
polynomial antichain growth} if for every antichain family $L'\subseteq L$ we
have that $L'$ is polynomial.\end{definition}

Note that we could have chosen to define
exponential antichain growth as containing an exponential antichain (rather than
an exponential antichain family).  We will eventually see (Corollary
\ref{cor:defequiv}) that for regular languages the two notions are equivalent.
However, for general languages they are not; indeed the following proposition shows
that the two possible definitions are not equivalent even for context-free
languages.

\begin{proposition}\label{prop:antfam}There exists a context-free language $L$ such
that $L$ has exponential antichain growth but all antichains in $L$ are
finite.
\end{proposition}
\begin{proof}Let $\Sigma=\{a,b,0,1\}$ with $\prec\,= \{(a,b)\}$.  Let
\[L=\bigcup_{n=1}^\infty L_n = \bigcup_{n=1}^\infty a^{n-1}b\{0,1\}^n.\]
Then each $L_n$ is an antichain of size $2^n$ consisting of words of length $2n$, 
but we have $L_1>L_2>L_3>\ldots$ so any antichain is a subset of $L_k$ for
some $k$ and hence is finite (the notation $L_1>L_2$ means that for any $w_1\in
L_1$ and $w_2\in L_2$ we have $w_2\prec w_1$). Plainly $L$ is a context-free 
language.  
\end{proof}

We observed above that Kleene star does not preserve the property of being an
antichain.  We conclude this section by establishing Lemma \ref{lem:expquas}, which addresses
this problem; if our goal is to find a large antichain, it suffices to find a
large quasiantichain (where the precise meaning of `large' is having exponential
growth).

As a preliminary, we observe the straightforward fact that taking finite unions does not
change the polynomial or exponential growth character of languages.

\begin{lemmarep}\label{lem:finiteunion}Let $L_1,L_2,\ldots,L_k$ be languages, such
that $\bigcup_{i=1}^k L_i$ has exponential growth of order $\epsilon$ (respectively super-polynomial
growth).  Then $L_i$ has exponential growth of order $\epsilon$ (respectively
super-polynomial growth) for
some $i$.\end{lemmarep}
\begin{proof}
Suppose that $L=\cup_{i=1}^k L_i$ has exponential growth of order $\epsilon$.
Then for any $\epsilon' < \epsilon$ we have 
\[0 < \limsup_{n\rightarrow \infty} \frac{|L|_{=n}}{2^{\epsilon' n}} \leq \sum_{i=1}^k
\limsup_{n\rightarrow \infty}\frac{\left|L_i\right|_{=n}}{2^{\epsilon' n}},\]
and hence we have $\limsup_{n\rightarrow \infty}
\frac{\left|L_i\right|_{=n}}{2^{\epsilon' n}} > 0$ for some $i$.

Similarly, suppose that $L$ has growth beyond all polynomial orders.  Then for
every $m$ we have
\[\infty = \limsup_{n\rightarrow \infty} \frac{|L|_{=n}}{n^m} \geq
\max_{i=1,\ldots,k} \limsup_{n\rightarrow \infty}\frac{\left|L_i\right|_{=n}}{
n^m},\]
and hence there is some $i_m$ such that $\limsup_{n\rightarrow
\infty}\frac{\left|L_{i_m}\right|_{=n}}{n^m} =\infty$.  Now by the pigeon-hole
principle there must be some $i$ such that $i=i_m$ for arbitrarily large $m$,
and so $L_i$ has growth beyond all polynomial orders.
\end{proof}

We are now ready to prove Lemma \ref{lem:expquas}.  
We do this by constructing
an exponential prefix-free subset of the exponential quasiantichain, which will
therefore be an exponential antichain.  We do this by a Ramsey-style 
argument:
always maintaining the invariant of exponential growth, at each step we pick a
fixed word $w$ of length $k$, throw away that word if it is in the set, and also
throw away all longer words of which $w$ is \emph{not} a prefix.  We will see
that by Lemma \ref{lem:finiteunion} it is always possible to choose $w$ such
that this process preserves the invariant. 

\begin{lemmarep}\label{lem:expquas}Let $L$ be an exponential quasiantichain.  Then there exists an
exponential antichain $L'\subseteq L$.\end{lemmarep}
\begin{proof}
Suppose that $L$ has exponential growth, that is that
$|L|_{=n}>2^{\epsilon n}$ infinitely often for some $\epsilon$.
We will construct a prefix-free set $S\subset \Sigma^*$ such that
$S\cap L$ has exponential growth.  We will construct a sequence of sets
$S_0\supseteq S_1 \supseteq S_2\supseteq\ldots$ (and associated integers $n_0 <
n_1 < n_2 < \ldots$ and reals $\epsilon_0 > \epsilon_1 > \epsilon_2 > \ldots >
\epsilon'$ for initially chosen $0<\epsilon' < \epsilon$) such that the 
intersection of the $S_i$ is the desired set $S$.  
In particular we will maintain the invariant that
each $S_i\cap L$ has $|S_i\cap L|_{=n}>2^{\epsilon_i n}$ infinitely often.

Let $S_0=\Sigma^*$ and let $n_0 = 0$.  To produce $S_{i+1}$, note that by the
invariant we can choose some $n=n_{i+1}>n_i$ such that $|S_i\cap
L|_{=n}>2^{\epsilon_i n}$.   Now $S_i\cap L$ has exponential growth of order
$\epsilon_i$, hence so does $(S_i\cap L)_{>n}$.  Now 
\[(S_i\cap L)_{>n} = \bigcup_{w\in\Sigma^n}(S_i\cap L)\cap w\Sigma^+,\] which is
a finite union.  Hence by Lemma \ref{lem:finiteunion} we have that $(S_i\cap
L)\cap w\Sigma^+$ has exponential growth of order $\epsilon_i$ for some
$w=w_{i+1}\in\Sigma^n$.  Thus taking any $\epsilon_{i+1}$ with
$\epsilon'<\epsilon_{i+1}<\epsilon_i$ we have that $|(S_i\cap L)\cap
w_{i+1}\Sigma^+|_{=n}>2^{\epsilon_{i+1}n}$ infinitely often.  Now let 
\[S_{i+1} = S_i \cap \left( \Sigma^{\leq n_i} \cup \left(\Sigma^n\setminus
w_{i+1}
\right) \cup w_{i+1}\Sigma^+\right).\]
Informally, to form $S_{i+1}$ we leave intact the part of $S_i$ consisting of
words of length $n_i$ or shorter.  To this we add all the words of length $n$
in $S_i$ apart from $w_{i+1}$, and all the words of length $>n$ which have
$w_{i+1}$ as a
prefix.  Since $S_i\cap w_{i+1}\Sigma^+ \subseteq S_{i+1}$ we clearly preserve the
exponential growth invariant.

We must now show that $S$ is prefix free and that it has exponential
intersection with $L$.  Note that the set of word lengths in $S$ is $\{n_0,
n_1, n_2,\ldots\}$, and also that \[S_{=n_i}=\left(S_i\right)_{=n_i}.\]  So
\[\begin{aligned}
|S\cap L|_{=n_i} &= \left|S_i\cap L\right|_{=n_i} \\
&\geq \left|S_{i-1}\cap
L\right|_{=n_i}-1 \\
&> 2^{\epsilon_{i-1}n} -1\\
&> 2^{\epsilon' n}-1,
\end{aligned}\]
where the first inequality is by the construction of $S_i$ from $S_{i-1}$ (up to
a single word of length $n_i$ is removed, namely $w_i$), the second is by the
definition of $n_i$ and the third is by the definition of $\epsilon_{i-1}$.
Hence $S\cap L$ has exponential growth of order at least $\epsilon'$.

To show that $S$ is prefix free, we show that $S_i$ has no pair $w< w'$ such
that $|w|=n_i$.  Indeed, by the definition of $S_i$ we must have on the one hand
that $w\neq w_i$ but on the other that $w'\in w_i\Sigma^+$, and so $w\not\leq w'$.
Since $S\subseteq S_i$ for all $i$ and $S$ only contains words of length $n_i$
for some $i$, we have that $S$ is prefix-free.
\end{proof}

\section{Regular languages}\label{sec:reg}

The dichotomy between polynomial and exponential language growth for regular
languages has been independently discovered at least six times (see citations in
\cite{gawrychowski2008finding}), in each case based on the fact that a regular language $L$ has
polynomial growth if and only if $L$ is \emph{bounded} (that is, $L\subseteq
w_1^*\ldots w_k^*$ for some $w_1,\ldots,w_k$); otherwise $L$ has exponential
growth.

In \cite{gawrychowski2008finding}, Gawrychowski, Krieger, Rampersad and Shallit
describe a polynomial time algorithm 
for determining whether a language is bounded.  The key idea is to consider the
sets $L_q$ of words which can be generated beginning and ending at state $q$.
$L$ is bounded if and only if for every $q$ we have that $L_q$ is
\emph{commutative} (that is, that $L_q\subseteq w^*$ for some
$w$), and this can be checked in polynomial time.

In this section, we generalise this idea to the problem of antichain growth by
showing that $L$ has polynomial antichain growth if and only if $L_q$ is a chain for
every $q$, and otherwise $L$ has exponential antichain growth.  This is sufficient to
establish the dichotomy theorem (Theorem \ref{thm:dichotomy}).  To give an
algorithm for distinguishing the two cases (Theorem \ref{thm:regalg}), we show
how to produce an automaton whose language is empty if and only if $L_q$ is a
chain (roughly speaking the automaton accepts pairs of incomparable words in
$L_q$).  

Before proving the main theorems, we first establish (Lemma \ref{lem:prodpoly})
that if $L_1$ and $L_2$ have polynomial antichain growth then so does $L_1L_2$. 
Moreover if the rates of polynomial growth of $L_1$ and $L_2$ are at most $k_1$ 
and $k_2$ respectively then the rate of polynomial growth of $L_1L_2$ is at most
$k_1+k_2+1$.

\begin{lemmarep}\label{lem:prodpoly}Let $L_1,L_2$ be languages with polynomial
antichain growth of order at most $k_1$ and $k_2$ respectively.  Then
$L_1L_2$ has polynomial antichain growth of order at most $k_1+k_2+1$.
\end{lemmarep}
\begin{proof}
Let $C_1,C_2$ be such that for any antichain family $L\subseteq L_i$ we have
$\left|L\right|_n < C_i n^{k_i}$ for all $n$.  We have 
\[\left(L_1L_2\right)_{=n} = \bigcup_{i=0}^n \left(L_1\right)_{=i}
\left(L_2\right)_{=n-i},\]
and so it suffices to prove that each $\left(L_1\right)_{=i}
\left(L_2\right)_{=n-i}$ contains antichains of size at most proportional to
$n^{k_1+k_2}$.

Let $L\subseteq \left(L_1\right)_{=i} \left(L_2\right)_{=n-i}$ be an antichain.
Then by the concatenation lemma we have that $\left\{w\in \left(L_1\right)_{=i}
\middle| ww'\in L \text{ for some $w'$}\right\}$ is an antichain, and hence it
has size at most $C_1i^{k_1}$.  On the other hand, by the prefixing lemma we
have that the set $\left\{w' \in \left(L_2\right)_{=n-i} \middle| ww'\in
L\right\}$ is an antichain for each $w$, and hence it has size at most $C_2n^{k_2}$.  Since
\[L = \bigcup_{w\in \left(L_1\right)_{=i}} \left\{ww' \middle| w'\in
\left(L_2\right)_{=n-i}, ww'\in L\right\},\]
we have that 
\[\begin{aligned}|L| &\leq \left|\left\{w\in \left(L_1\right)_{=i}
\middle| ww'\in L \text{ for some $w'$}\right\}\right| \times \max_{w} 
\left|\left\{w' \in \left(L_2\right)_{=n-i} \middle| ww'\in
L\right\}\right|\\
&\leq C_1n^{k_1} C_2n^{k_2}\\
&= C_1C_2n^{k_1+k_2},
\end{aligned}\]
as required.
\end{proof}

We are now ready to prove the main theorem, generalising the condition for 
polynomial language growth (that $L_q$ is commutative for every $q$) to one for 
polynomial antichain growth: that $L_q$ is a chain for every relevant $q$.

\begin{definition}A state $q$ of an automaton $\Aa = (Q,\Sigma,\Delta,q_0,F)$ is
\emph{accessible} if $q$ is reachable from $q_0$ and \emph{co-accessible} if $F$
is reachable from $q$.
\end{definition}
\begin{definition}Let $\Aa = (Q,\Sigma,\Delta,q_0,F)$ be an NFA.  Then for each
$q_1,q_2\in Q$, the automaton
$\Aa_{q_1,q_2}\triangleq(Q,\Sigma,\Delta,q_1,\{q_2\})$.\end{definition}

\begin{theorem}\label{thm:dichotomy}
Let $\Aa = (Q,\Sigma,\Delta,q_0,F)$ be an NFA over a partially ordered alphabet.  
Then \begin{enumerate}[(i)]\item $\LL(\Aa)$ has polynomial antichain growth if and only if $\LL(\Aa_{q,q})$ is a chain
for every accessible and co-accessible state $q$, and \item if $\LL(\Aa)$ does not have polynomial
antichain growth then it contains an exponential antichain (and hence has
exponential antichain growth).\end{enumerate}
\end{theorem}
\begin{proof}
Suppose that $w_1,w_2\in\LL(\Aa_{q,q})$ with $w_1\not\sim w_2$ and $q$
accessible and co-accessible, so $w\in\LL(\Aa_{q_0,q})$ and $w'\in\LL(\Aa_{q,q'})$ for some
$w,w'$ and some $q'\in F$.  Now by the Kleene star Lemma we have that $(w_1+w_2)^*$ is an exponential quasiantichain
and so by Lemma \ref{lem:expquas} there is an exponential antichain $L'\subseteq
(w_1+w_2)^*$.  Then by the Prefixing and Postfixing Lemmas we have that
$wL'w'\subseteq L$ is an exponential antichain.

For the converse, we proceed by induction on $|Q|$.  
Let $Q'=Q\setminus\{q_0\}, F'=F\setminus\{q_0\}$ and
$\Delta'(q,a)=\Delta(q,a)\setminus\{q_0\}$ for all $q\in Q', a\in \Sigma$.  For any
$q\in Q'$, let $\Aa'_q=(Q',\Sigma,\Delta',q,F')$.  Then by
the inductive hypothesis we have that $\LL(\Aa'_q)$ has polynomial antichain
growth.  Also, since $L_{q_0} = \LL(\Aa_{q_0,q_0})$ is a 
chain it has polynomial (in
particular constant) antichain growth.  Now we have 
\[\LL(\Aa)\subseteq L_{q_0} \cup \bigcup_{q\in
Q'}\bigcup_{a\in\Delta(q_0,q)} L_{q_0}a\LL(\Aa'_q).\]

By Lemma \ref{lem:prodpoly}, each $L_{q_0}a\LL(\Aa'_q)$ also has polynomial
antichain growth, and hence by Lemma \ref{lem:finiteunion} so does the finite
union.
\end{proof}

A trivial restatement of part (ii) of the theorem shows that the two possible
definitions of antichain growth are equivalent.

\begin{corollary}\label{cor:defequiv}Let $L$ be a regular language.  Then $L$ has exponential
(respectively super-polynomial) antichain growth if and only if $L$ contains an
exponential (respectively super-polynomial) antichain.\end{corollary}

Using Theorem \ref{thm:dichotomy} we can produce an algorithm
for distinguishing the two cases.

\begin{theoremrep}\label{thm:regalg}There exists a polynomial time algorithm to determine whether the
language of a given NFA $\Aa$ has exponential antichain growth.\end{theoremrep}
\begin{proof}
First remove all states which are not accessible and co-accessible (trivial
flood fill: for instance, to compute the set of accessible states, initialise
the set $X=\{q_0\}$ and then repeatedly add states to $X$ if they can be reached
by a transition from a state in $X$), to give
$\Aa=(Q,\Sigma,\Delta,q_0,F)$.  We will now check
for each state $q$ whether $\LL(\Aa_{q,q})$ is a chain.

Let $\Sigma'$ denote the alphabet $\{x'|x\in \Sigma\}$ (that is, an alphabet of
fresh letters of the same size as $\Sigma$).  Let $\Aa'$ be the automaton
corresponding to $\Aa$ over $\Sigma'$.  Let $\Bb=(\Sigma\cup
\{s_0,s_1\},\Sigma\cup\Sigma',\widetilde{\Delta},s_0,\{s_1\})$ be an NFA, where
$s_0,s_1$ are fresh and $\widetilde{\Delta}$ is given by (for all $a\in\Sigma$):
$\widetilde{\Delta}(s_0,a)=\{a\}, \widetilde{\Delta}(a,a')=\{s_0\}$,
$\widetilde{\Delta}(a,b')=\{s_1\}\text{ for all $b$ with
$a\not\preceq b$ and $b\not\preceq a$}$, 
$\widetilde{\Delta}(s_1,a)=\widetilde{\Delta}(s_1,a')=\{s_1\}$,
and all other sets empty.

Then $\Bb$ has two important properties.  Firstly every word accepted by $\Bb$
is a shuffle of two words $w_1$ and $w_2'$, where $w_1,w_2\in\Sigma^*$ such that
$w_1\not\sim w_2$ and $w_2'$ is $w_2$ over the primed alphabet
(intuitively, the two words are equal for the part where $s_0$ is visited, and then they
first differ by two incomparable letters).  Secondly, for every $w_1\not\sim
w_2$ we have that the perfect shuffle of $w_1$ and $w_2$ is accepted by $\Bb$
(that is, if $w_1=a_1a_2\ldots a_k, w_2=b_1b_2\ldots b_{k'}$ and WLOG $k<k'$
then $a_1b_1'a_2b_2'\ldots a_kb_k'b_{k+1}'\ldots b_{k'}'$ is accepted by $\Bb$).

Hence $\LL(\Aa_{q,q})$ is a chain if and only if
$\LL((\Aa_{q,q}\interleave\Aa'_{q,q})\cap \Bb)$ is empty, which can be checked in
polynomial time (where $\interleave$ is the interleaving operator, which can be
realised by a product construction). Note that in fact it suffices to check a
single representative of each strongly connected component of $\A$.
\end{proof}

\section{Precise growth rates}\label{sec:precisegrowth}

In \cite{gawrychowski2008finding} the authors give an algorithm to compute the order of polynomial language growth for the language of a given NFA; on the other hand efficiently computing the order of exponential growth is an open problem.  In this section we give an algorithm to compute the order of polynomial antichain growth for the language of a given NFA.  We do this by first giving an algorithm for DFA, and then showing that in fact it also works for NFA.  We will assume throughout without loss of generality that all states are accessible and co-accessible.

\begin{definition}
Let $\A = (Q,q_0,F,\Sigma,\delta)$ be a DFA over a partially ordered alphabet.
Let $G_{\A}=(Q,E)$ be the directed graph with vertex-set $Q$ such that
$(q,q')\in E$ if and only if $q\xrightarrow{w} q'$ for some $w\in \Sigma^*$.

Let $G'_{\A}=(Q,E')$ be the directed graph with $(q,q')\in E'$ if and only if
there exist words $w\not\sim w'\in \Sigma^*$ such that $q\xrightarrow{w} q$ and
$q\xrightarrow{w'} q'$.  We will write $L_{q,q'}\triangleq \LL(\A_{q,q'})$.
\end{definition}

We will generally omit the subscript $\A$s from now on, where this will not
cause confusion.

Note that by Theorem \ref{thm:dichotomy}, we have that $G'$ is a directed
acyclic graph (DAG) if and only if $\LL(\A)$ has polynomial antichain growth. By
a similar argument to the proof of Theorem \ref{thm:regalg}, the graph $G'$ can
be computed in polynomial time. Clearly $G$ can be computed in polynomial time
using a flood fill.

\begin{definition}\label{def:DP}
Let $\A=(Q,q_0,F,\Sigma,\delta)$ be a DFA with polynomial antichain growth.  For
a directed path $P=q_0q_1\ldots q_l$ (not necessarily simple) in
$G_{\A}$, let
\[D(P) = \left|\left\{i\in \{0,\ldots,l-1\}\middle| (q_i,q_{i+1})\in
E(G'_{\A})\right\}\right| + \begin{cases}1\text{ if $|L_{q_m,q_l}|=\infty$} \\
0 \text{ otherwise.}\end{cases},\]
where $m=\max \{i+1|(q_i,q_{i+1})\in G'_\A\}$ if this exists, and 0 otherwise.
\end{definition}

Observe that if $|L_{q_m,q_l}|=\infty$ then we have $ww'^*w''\subseteq
L_{q_m,q_l}$ for some $w,w',w''$.

\begin{lemma}\label{lem:DAcomp}
Let $\A=(Q,q_0,F,\Sigma,\delta)$ be a DFA with polynomial antichain growth.  
Let $\mathcal{P}$ be the set of directed paths from $q_0$ to an element of $F$.  
Then the quantity 
\[D_\A=\max_{P\in\mathcal{P}} D(P)\]
is well-defined and can be computed in polynomial time.
\end{lemma}
\begin{proof}
To show that $D_\A$ is well-defined, observe that no directed cycle in $G$
contains an edge in $G'$.  Indeed, suppose that $q_1q_2\ldots q_1$ is a directed
cycle in $G$, with $(q_1,q_2)\in E(G')$.  Then we have $q_1\xrightarrow{w} q_1$
and $q_1\xrightarrow{w'} q_2$ for some $w\not\sim w'\in \Sigma^*$.  Also we have
$q_2\xrightarrow{w''} q_1$ for some $w''\in \Sigma^*$.  But then
$q_1\xrightarrow{w'w''}q_1$ and $w'w''\not\sim w$ by the Concatenation Lemma,
contradicting polynomial antichain growth of $\LL(\A)$.  Hence $D(P)$ is
bounded.

For a polynomial time algorithm, first expand $G$ and $G'$ by adding a sink vertex $v_f$ for each $f\in
F$.  For each $q$ such that $|L_{q,f}|=\infty$ put $(q,v_f)\in E(G)$ and $(q,v_f)\in
E(G')$.   Then add a further vertex $v$ with $(f,v)\in E(G)$ and $(v_f,v)\in
E(G)$ for all $f\in F$.  Then $D_\A$ is precisely the maximum
number of edges of $G'$ contained in a directed path from $q_0$ to $v$ in $G$.

Form the graph $G''$ on vertex-set
$Q\cup\{v\}$ by $(v_1,v_2)\in E(G'')$ if and only if there is a path from $v_1$
to $v_2$ in $G$ containing a single edge of $G'$.  Then we have that $G''$ is a DAG
(by the first observation), and $D_\A$ is the longest path from $q_0$ to $v$ in
$G''$, which can be found by a simple dynamic programming algorithm.
\end{proof}

We will show that the order of polynomial antichain growth of $\LL(\A)$ is 
precisely $D_\A-1$.

\begin{lemma}\label{lem:lowerbound}
Let $\A=(Q,q_0,F,\Sigma,\delta)$ be a DFA with polynomial antichain growth.
Then $\LL(\A)$ has polynomial antichain growth of order at least $D_\A-1$.
\end{lemma}

\begin{proof}
Let $P=q_0q_1\ldots q_l$ be a path with
$D(P)=D_\A$.  Let $i_1,\ldots,i_k$ be such that $(q_{i_j},q_{i_j+1})\in
E(G'_\A)$ for all $j$.  Let $w_1,\ldots,w_k,w'_1\ldots,w'_k,w \in\Sigma^*$ be
such that $w_j\not\sim w'_j$ for all $j$, $q_{i_j}\xrightarrow{w_j} q_{i_j}$ for
all $j$, $q_{i_j}\xrightarrow{w'_j} q_{i_{j+1}}$ for all $j<k$,
$q_{i_k}\xrightarrow{w'_k} q_l$, and $q_0\xrightarrow{w} q_{i_1}$.

Suppose that $|L_{q_m,q_l}|=\infty$ (with $m=i_k$ defined as in Definition
\ref{def:DP}), and let $w',w'',w'''\in\Sigma^*$ be such
that $w'w''^*w'''\subseteq L_{q_m,q_l}$.  Then $L=ww_1^*w_1'w_2^*w_2'\ldots
w_k^*w'w''^*w'''$ is an
antichain family with polynomial growth of order $k=D_\A-1$.  Similarly if
$|L_{q_m,q_l}|<\infty$, then $L=ww_1^*w_1'w_2^*w'_2\ldots w_k^*w'_k$ is an antichain
with polynomial growth of order $k-1=D_\A-1$.
\end{proof}

We will now prove the upper bound.  Our strategy will be to classify words by
the edges of $G'$ they visit.  We first show a preliminary lemma, which
bounds the antichain growth from regions between edges of $G'$.

\begin{lemma}\label{lem:blackruns}
Let $q_1,q_2\in Q$, and let $L\subseteq L_{q_1,q_2}$ be the set of words such that no
edges of $G'$ appear in the runs corresponding to elements of $L$.  Then $L$ has
antichain growth of order at most 0.
\end{lemma}
\begin{proof}
Without loss of generality we may assume that $\A$ does not have any transitions
labelled by more than a single letter (by introducing additional states if
necessary; in particular we can set $Q'=Q\times \Sigma$ and ensure that
$\delta'(q,x) \in Q\times\{x\}$ for all $x\in \Sigma$).

We will show that $L$ cannot contain two incomparable words that correspond
after removal of loops to the same sets of simple paths in $G$.\footnote{Note
that since removal of loops may be done in many different ways, a single path
may correspond to multiple simple paths.  We are asserting that $L$ cannot
contain two incomparable words which correspond to precisely the same
\emph{sets} of simple paths.}  Since $G$ is
finite and hence contains only finitely many simple paths, this suffices to
establish the result.

Suppose that $w_1\not\sim w_2$ correspond to the same simple path $P$.  Suppose
that the first point of divergence of $w_1$ and $w_2$ is at state $q$; that is,
that $w_1=wx_1w_1'$ and $w_2=wx_2w_2'$ with $x_1\neq x_2\in \Sigma$ and
$q_1\xrightarrow{w} q$ (see Figure \ref{fig:blackrun}).  Without loss of generality we may assume that $q$ and
$\delta(q,x_1)$ lie on $P$.  

Since the path for $w_2$ corresponds to $P$ after removal of cycles, we must
have that $w_2'=w_2''w_2'''$ with $q\xrightarrow{x_2w_2''} q$ and
$q\xrightarrow{w_2'''}q_2$.  But $w_1\not\sim w_2$ and $x_1\neq x_2$ so
$x_1\not\sim x_2$ and so $x_1\not\sim x_2w_2''$.  Hence $(q,\delta(q,x_1))\in
G'$, which is a contradiction.
\end{proof}

\begin{figure}[!htb]
\centering
\includegraphics[width=0.5\textwidth]{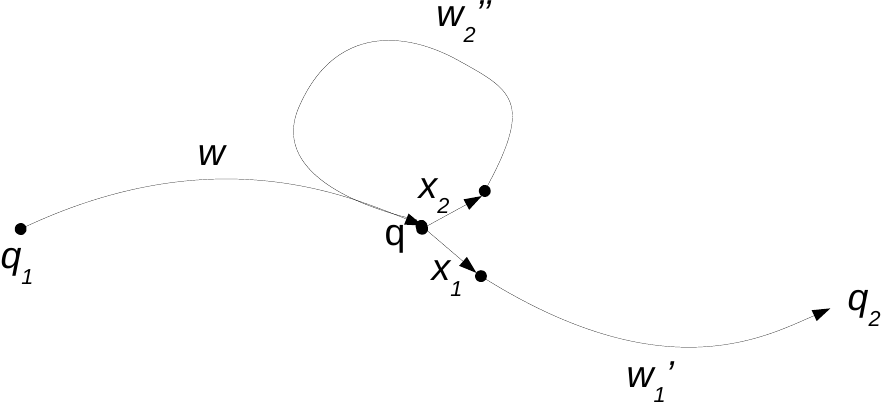}
\caption{The proof of Lemma \ref{lem:blackruns}}\label{fig:blackrun}
\end{figure}

\begin{lemma}\label{lem:upperbound}
Let $\A=(Q,q_0,F,\Sigma,\delta)$ be a DFA with polynomial antichain growth.
Then $\LL(\A)$ has polynomial antichain growth of order at most $D_\A-1$.
\end{lemma}

\begin{proof}
We may assume without loss of generality that there is only a single accepting
state, say $q_f$ (otherwise consider seperately the automata $\A_1,\ldots,
\A_{|F|}$ which agree with $\A$ except for having only a single accepting state; 
then on the one hand we have $D_{\A}=\max D_{\A_i}$, but on the other hand 
$\LL(\A) = \bigcup \A_i$ which is a finite union and hence the order of antichain growth 
of $\LL(\A)$ is the maximum of the orders of growth of the $\LL(\A_i)$).

We classify words by the edges of $G'$ that appear in their accepting runs.  We
shall show that the set of words corresponding to a fixed sequence $P$ of $G'$-edges
has antichain growth of order at most $D(P)$ (where $D(P)=|P|-1$ or $|P|$
depending on whether the set of accepted words beginning at the last vertex of
$P$ is finite).  Since the number of relevant 
$G'$-edge sequences is finite (recalling that no edge of $G'$ is contained in a
directed cycle in $G$ and so no $G'$-edge can appear more than once), this will
suffice to establish the result.

Let $(q_1,q_1'),\ldots,(q_k,q_k')$ be a set of $G'$-edges.  Then the set $L$ of
words which have this sequence of $G'$-edges in their run is given by
\[L=L'_{q_0,q_1}X_1L'_{q_1',q_2}X_2L'_{q_2',q_3} \ldots X_k L'_{q_k',q_f},\]
where $X_i = \left\{x\in \Sigma \mid \delta(q_i,x)=q_i'\right\}$ and
$L'_{q,q'}\subset L_{q,q'}$ is the set of words whose runs do not include edges
of $G'$.

The $X_i$ are finite and hence have antichain growth of order $-1$.  By Lemma
\ref{lem:blackruns} the $L'_{q_i',q_{i+1}}$ and also $L'_{q_0,q_1}$ and
$L'_{q_k',q_f}$ have antichain growth of order at most 0.  Moreover if
$L_{q_k',q_f}$ is finite then so is $L'_{q_k',q_f}\subseteq L_{q_k',q_f}$ and so
it has antichain growth of order $-1$.  The result follows by Lemma
\ref{lem:prodpoly}.
\end{proof}

Combining Lemmas \ref{lem:DAcomp}, \ref{lem:lowerbound} and \ref{lem:upperbound} yields

\begin{theorem}\label{thm:DFAgrowth}Let $\A=(Q,q_0,F,\Sigma,\delta)$ be a DFA with polynomial
antichain growth.  Then $\LL(\A)$ has polynomial antichain growth of order exactly
$D_\A-1$, which can be computed in polynomial time.
\end{theorem}

We now show how to extend this algorithm to the case of NFA.  Note that $D_\A$
as defined above is well-defined for NFA just as for DFA, and that the algorithm
to compute it in polynomial time is equally applicable.  It therefore remains to
show that for NFA we also have that if $\A$ has polynomial antichain growth then
it has antichain growth of order exactly $D_\A-1$.

We do this by showing (Lemma \ref{lem:DAinvar}) that $D_\A$ depends only on
the language $\LL(\A)$, so that if $\A$ and $\A'$ are NFA with
$\LL(\A)=\LL(\A')$ then $D_\A=D_{\A'}$.  Having shown this we then consider $\A'$
to be the determinisation of $\A$.  This is a DFA with $\LL(\A')=\LL(\A)$, and by
Theorem \ref{thm:DFAgrowth} we have that $\LL(\A')$ has polynomial antichain
growth of order $D_{\A'}-1=D_{\A}-1$.

We will first show (Lemma \ref{lem:nfaoffset}) that if
$L=v_0w_1^*v_1w_2^*v_2\ldots w_k^*v_k \subseteq \LL(\A)$ then there exists a
single sequence of states $q_1,q_2,\ldots,q_k$ which essentially realises $L$
(that is, up to various offsets we have $v_i\in
\LL(\A_{q_i,q_{i+1}})$ and $w_i^*\in \LL(\A_{q_i,q_i})$).

\begin{lemma}\label{lem:nfaoffset}Let $\A=(Q,q_0,F,\Sigma,\Delta)$ be an NFA such that
$v_0w_1^*v_1w_2^*v_2\ldots w_k^*v_k \subseteq \LL(\A)$.  Then then there exists
a sequence of states $q_1,q_2,\ldots,q_{k+1}$ and integers $m_1,m_2,\ldots m_k$,
$m_1',m_2',\ldots, m_k'$
and $n_1,n_2,\ldots,n_k$ such that 
\begin{enumerate}[(i)]
\item $v_0w_1^{m_1} \in \LL(\A_{q_0,q_1})$ and $w_k^{m'_k}v_k \in
\LL(\A_{q_k,F})$,
\item for all $0 < i < k$ we have $w_i^{m'_i} v_i w_{i+1}^{m_{i+1}}
\in \LL(\A_{q_i,q_{i+1}})$, and
\item for all $0 < i \leq k$ we have
$w_i^{n_i} \in \LL(\A_{q_i,q_i})$.
\end{enumerate}
\end{lemma}
\begin{proof}
Consider an accepting run for $v_0w_1^{|Q|+1}v_1w_2^{|Q|+1}v_2 \ldots
w_k^{|Q|+1}v_k \in \LL(\A)$, and write $q(s)$ for the state reached in this run
after the word $s$.  By the pigeon-hole principle, we must have
$q(v_0w_1^{m_1})=q(v_0w^{m_1+n_1})=q_1$ (say) for some $m_1\geq 0$ and some $n_1 > 0$ with 
$m_1+n_1\leq |Q|+1$.  Let $m_1'=|Q|+1-m_1-n_1$.  Similarly for each $i$ we have
$q(v_1w_1^{|Q|+1}v_2\ldots w_i^{m_i})=q(v_1w_1^{|Q|+1}v_2\ldots w_i^{m_i+n_i}) =
q_i$ (say) for some $m_i\geq 0$ and $n_i > 0$ with $m_i+n_i\leq |Q|+1$.  Let
$m'_i=|Q|+1-m_i-n_i$.  Then these $q_i,m_i,m'_i$ and $n_i$ give the result.
\end{proof}

\begin{lemma}\label{lem:DAinvar}
Let $\A$ and $\A'$ be NFA with $\LL(\A)=\LL(\A')$.  Then $D_{\A}=D_{\A'}$.
\end{lemma}
\begin{proof}
Let $\A=(Q,q_0,F,\Sigma,\Delta)$ and $\A'=(Q',q_0',F',\Sigma,\Delta')$.

Suppose that $D_{\A'} = k$.  Then by an identical argument to the proof of Lemma
\ref{lem:lowerbound} we have that $v_0w_1^*v_1w_2^*v_2 \ldots w_k^* v_k\subseteq
\LL(\A')=\LL(\A)$ for some $v_0,\ldots,v_k,w_1,\ldots,w_k\in \Sigma^*$ with $w_i\not\sim
v_i$.  Then by Lemma \ref{lem:nfaoffset} there exists a sequence of states
$q_1,q_2,\ldots,q_{k+1} \in Q$ and integers $m_1,m_2,\ldots,m_k, m'_1,m'_2,
\ldots m'_k$ and
$n_1,n_2,\ldots,n_k$ such that (i)--(iii) in the statement of the lemma hold.
Now since $w_i\not \sim v_i$ we have $w_i^{k_in_i} \not\sim
w_i^{m'_i}v_iw_{i+1}^{m_{i+1}}$ for sufficiently large $k_i$ and so 
\[D_{\A}\geq k = D_{\A'}.\]

Similarly $D_{\A'}\geq D_{\A}$, and hence $D_{\A}=D_{\A'}$.
\end{proof}

\begin{theorem}\label{thm:nfapoly}
Let $\A$ be an NFA with polynomial antichain growth.  Then $\LL(\A)$ has
polynomial antichain growth of order exactly $D_{\A}-1$.
\end{theorem}
\begin{proof}
Let $\A'$ be the powerset determinisation of $\A$, so $\A'$ is a DFA with
$\LL(\A')=\LL(\A)$.  By Theorem \ref{thm:DFAgrowth}, $\LL(\A')$ has polynomial
antichain growth of order exactly $D_{\A'}-1$, and by Lemma \ref{lem:DAinvar}
we have $D_{\A'}=D_{\A}$.
\end{proof}

\section{Context-free languages}\label{sec:contextfree}

In \cite{ginsburg1964bounded}, Ginsburg and Spanier show (Theorem 5.1) that a
context-free grammar $G$ generates a bounded language if and only if the sets
$L_A(G)$ and $R_A(G)$ are commutative for all non-terminals $A$, where $L_A$ and
$R_A$ are respectively the sets of possible $w$ and $u$ in productions $A\prodstar
wAu$.  They also give an algorithm to decide this (which
\cite{gawrychowski2008finding}
improves to be in polynomial time).

We generalise this to our problem by showing that $G$ generates a language with
polynomial antichain growth if and only $L_A(G)$ and also the sets $R_{A,w}(G)$ of
possible $u$ for each fixed $w$ are chains, and that otherwise $\LL(G)$ has
exponential antichain growth.  However, we will show that the problem of
distinguishing the two cases is undecidable, by reduction from the CFG
intersection emptiness problem.

Except where otherwise specified, we will assume all CFGs have starting symbol $S$ and that all nonterminals are
accessible and co-accessible: for any nonterminal $A$ we have $S\prodstar
uAu'$ for some $u,u'\in \Sigma^*$ and $A\prodstar v$ for some $v\in\Sigma^*$.

\begin{definition}Let $G$ be a context-free grammar (CFG) over $\Sigma$.  Then for any
nonterminal $A$ let 
\[L_A(G) = \{w\in \Sigma^*| \exists u\in \Sigma^*: A\prodstar
wAu\}.\]
\end{definition}

\begin{lemma}Let $G$ be a CFG over $\Sigma$ and $A$ some nonterminal such that
$L_A(G)$ is not a chain.  Then $\LL(G)$ contains an exponential antichain.
\end{lemma}
\begin{proof}
Since $L_A(G)$ is not a chain, we have $w_1,w_2,u_1,u_2$ with $w_1\not\sim w_2$
such that $A\prodstar w_1Au_1$ and $A\prodstar w_2Au_2$.  Now $A$ is accessible
and co-accessible
so also $S\prodstar uAu'$ and $A\prodstar v$ for some $u,u',v\in\Sigma^*$.

Hence \[uw_{i_1}w_{i_2}\ldots w_{i_k}vu_{i_k}u_{i_{k-1}}\ldots u_{i_1}u'\subseteq
\LL(G),\] for any $i_1i_2\ldots i_k\in \{1,2\}^*$.  Write
$\phi:(w_1+w_2)^*\rightarrow (u_1+u_2)^*$ for the map $w_{i_1}w_{i_2}\ldots w_{i_k}
\mapsto u_{i_k}u_{i_{k-1}}\ldots u_{i_1}$ (with any ambiguity resolved
arbitrarily).

Now $\{w_{i_1}w_{i_2}\ldots w_{i_k}|i_1\ldots i_k\in \{1,2\}^*\} = (w_1+w_2)^*$
is a quasiantichain by Lemma \ref{lem:kleene}, clearly it is exponential and
hence by Lemma \ref{lem:expquas} it contains an exponential antichain $L$.  By
the Concatenation Lemma we have that $L'=\{lv\phi(l)|l\in L\}$ is an antichain, and
it is exponential because there is a bijection between $L$ and $L'$ such that
the length of each word in $L'$ exceeds the length of the corresponding word in
$L$ by a factor of at most $\frac{|v|+\max(|u_1|,|u_2|)}{\min(|w_1|,|w_2|)}$.  By
the Prefixing and Postfixing Lemmas we have that $uL'u'\subseteq \LL(G)$ is an
exponential antichain.
\end{proof}

\begin{definition}Let $G$ be a CFG over $\Sigma$.  Then for any nonterminal $A$
and any $w\in\Sigma^*$, let
\[R_{A,w}(G) = \{u\in \Sigma^*| A\prodstar wAu\}.\]
\end{definition}

\begin{lemma}Let $G$ be a CFG over $\Sigma$, $A$ some nonterminal and
$w\in\Sigma^*$ such that $R_{A,w}(G)$ is not a chain.  Then $\LL(G)$ has
exponential antichain growth.
\end{lemma}
\begin{proof}
We have $v,w,u,u'\in\Sigma^*$ and $u_1\not\sim u_2\in\Sigma^*$ such that
$S\prodstar uAu'$, $A\prodstar v$, $A\prodstar wAu_1$ and $A\prodstar wAu_2$.
Let \[L_i=uw^{2i}v(u_1u_2 + u_2u_1)^iu'.\]  Then $L_i$ is an antichain and
$\bigcup_{i=1}^\infty L_i$ is an exponential antichain family.
\end{proof}

\begin{lemmarep}Let $G$ be a CFG over $\Sigma$ such that $L_A(G)$ and $R_{A,w}(G)$
are chains for all nonterminals $A$ and all $w\in\Sigma^*$.  Then $\LL(G)$ has
polynomial antichain growth.
\end{lemmarep}
\begin{proof}
We proceed by induction on the number of nonterminals which appear on the right
hand side of productions in $G$.  Let $A$ be a nonterminal, and let $G'$ be the
CFG obtained from $G$ by deleting all productions mentioning $A$ on the
right hand side and changing the starting symbol to $A$.  Let $L'=\LL(G')$.  Then
by the inductive hypothesis $L'$ has polynomial antichain growth; say any
antichain family $L\subseteq L'$ has $|L|_{\leq k} < Ck^N$ for some fixed $C,N$.
If $A$ is not the starting symbol, let $G''$ be the CFG obtained from $G$ by
deleting all productions mentioning $A$, and let $L''=\LL(G'')$ (otherwise let
$L'' =\emptyset$).  By the inductive hypothesis $L''$ also has polynomial
antichain growth.  Now
we have 
\[\LL(G)\subseteq L'' \cup \left(L_A(G)L'\bigcup_{w\in\Sigma^*} R_{A,w}\right).\]

By Lemma \ref{lem:finiteunion} it suffices to prove that
$\LL(G) \setminus L''$ has polynomial antichain growth.  

Let
$L\subseteq \LL(G) \setminus L''$ be an antichain family.  Now since $L_A(G)$ is a chain
and $L_{=k}$ is an antichain, and morever every element of $L_{=k}$ is in 
$wL'R_{A,w}$ for some $w$, we have 
\[L_{=k}\subseteq \bigcup_{i=0}^k w_i L' R_{A,w_i},\]
for some $w_0 < w_1 < w_2 <\ldots < w_k$ with $|w_k|=k$ (recall that $<$ is defined
on $\Sigma^*$ as meaning strict prefix).  

Since $R_{A,w_i}$ is a chain and $L_{=k}$ is an antichain we cannot have $w_ilu,
w_ilu'\in L_{=k}$ for any $l\in L'$ and $u\neq u'\in R_{A,w_i}$.  Hence for each
$i$ there exists some function $\phi$ and $\widetilde{L}\subseteq L'$ such that 
\[L_{=k}\cap w_iL'R_{A,w_i} = \{w_i l \phi(l)|l\in \widetilde{L}\}.\]
Now since $L_{=k}$ is an antichain we have that $\widetilde{L}$ is a
quasiantichain and in particular an antichain family, and since also
$\widetilde{L}\subseteq L'_{\leq k}$ we have that $|\widetilde{L}| < Ck^N$.
Hence \[|L_{=k}\cap w_iL'R_{A,w_i}| \leq |\widetilde{L}| < Ck^N,\]
and so 
\[|L_{=k}| < (k+1)Ck^N < Ck^{N+2}\]
for sufficiently large $k$.
\end{proof}

Combining these three lemmas gives:

\begin{theorem}\label{thm:cfdichot}Let $L$ be a context-free language.  Then either $L$ has
exponential antichain growth or $L$ has polynomial antichain growth.
\end{theorem}

It is a straightforward exercise to show that the ambiguity of an NFA 
(the maximum number of accepting paths corresponding to a given word) can be 
represented as the width of a suitable context-free language, and hence Theorem 
\ref{thm:cfdichot} implies the well-known result that the ambiguity of an NFA 
has either polynomial or exponential growth (see Theorem 4.1 of 
\cite{weber1991ambiguity}).

We now show that the problem of distinguishing the two cases of antichain growth 
is undecidable for context-free languages, by
reduction from the CFG intersection emptiness problem.  In fact, it is
undecidable even to determine whether a given CFG generates a chain.

\begin{definition}{\sc CFG-Intersection} is the problem of determining whether two given
CFGs have non-empty intersection.  {\sc CFG-Chain} is the problem of
determining whether the language generated by a given CFG is a chain.  {\sc
CFG-ExpAntichain} is the problem of determining whether the language
generated by a given CFG has exponential antichain growth.
\end{definition}

\begin{lemma}{\sc CFG-Intersection} is undecidable.\end{lemma}
\begin{proof}\cite{ginsburg1966mathematical}, Theorem 4.2.1.
\end{proof}

\begin{lemma}There is a polynomial time reduction from {\sc CFG-Intersection} to
{\sc CFG-Chain}.
\end{lemma}
\begin{proof}
Let $G_1,G_2$ be arbitrary CFGs over alphabet $\Sigma$.  Let
$\widetilde{\Sigma}=\Sigma\cup\{0,1\}$, with an arbitrary linear order on
$\Sigma$, and $\Sigma < 0, \Sigma < 1$ but $0$ and $1$ incomparable.  Let
$\widetilde{G}$ be a CFG such that 
\[\LL(\widetilde{G}) = (\LL(G_1)0)\cup (\LL(G_2)1)\]
(which can trivially be constructed with polynomial blowup).  Then
$\LL(\widetilde{G})$ is a chain if and only if $G_1\cap G_2=\emptyset$.
\end{proof}

\begin{lemma}\label{lem:pfkleene}Let $L$ be a prefix-free chain.  Then $L^*$ is a chain.
\end{lemma}
\begin{proof}
Let $lw\not\sim l'w'$ be a minimum-length counterexample with $l,l'\in L$ and
$w,w'\in L^*$.  By minimality and the Prefixing Lemma we have that $l\neq l'$.  
Then by the Concatenation Lemma since $L$ is prefix-free we have
that $l\not\sim l'$, which is a contradiction.
\end{proof}

\begin{lemma}There is a polynomial time reduction from {\sc CFG-Chain} to
{\sc CFG-ExpAntichain}.
\end{lemma}
\begin{proof}
Let $G$ be a CFG over a partially ordered alphabet $\Sigma$.  Let
$\widetilde{\Sigma} = \Sigma\cup\{0\}$, with $\Sigma<0$.  Let $\widetilde{G}$ be
a CFG such that 
\[\LL(\widetilde{G})=(\LL(G)0)^*.\]
We claim that $\LL(\widetilde{G})$ has exponential antichain growth if and only if
$\LL(G)$ is not a chain.  Indeed, suppose that $l_1\not\sim l_2\in \LL(G)$.  Then $l_10\not\sim l_20$ and so by
Lemmas \ref{lem:kleene} and $\ref{lem:expquas}$ we have that
$(l_10+l_20)^*\subseteq \LL(\widetilde{G})$ contains an exponential antichain.

Conversely, suppose that $\LL(G)$ is a chain.  Then $\LL(G)0$ is a prefix-free
chain and so by Lemma \ref{lem:pfkleene} we have that $\LL(\widetilde{G})$ is a
chain.
\end{proof}

Combining these lemmas gives:

\begin{theorem}\label{thm:undecidable}
The problems {\sc CFG-Chain} and {\sc CFG-ExpAntichain} are undecidable.
\end{theorem}

\section{Tree automata}\label{sec:tree}

In this section, we generalise the definition of the lexicographic ordering to
tree languages, and prove a trichotomy theorem: regular tree languages have
antichain growth which is either polynomial, exponential or doubly exponential.

Notation and definitions (other than for the lexicographic ordering) are taken
from \cite{tata2007}, to which the reader is referred for a more detailed
treatment.  

\begin{definition}Let $\FF$ be a finite set of function symbols of arity $\geq 0$, and
$\XX$ a set of variables.  Write $\FF_p$ for the set of function symbols of
arity $p$.  Let $T(\FF,\XX)$ be the set of terms over $\FF$ and
$\XX$.   Let $T(\FF)$ be the set of \emph{ground terms} over $\FF$, which is
also the set of \emph{ranked ordered trees} labelled by $\FF$ (with rank given
by arity as function symbols).  
\end{definition}

For example, the set of ordered binary trees is $T(\FF)$, where $\FF=\{f,g,c\}$
and $f$ has arity 2, $g$ arity 1 and $c$ arity 0.

Note that this generalises the definition of finite words over an alphabet $\Sigma$,
by taking $\FF=\Sigma\cup\{\epsilon\}$, giving each $a\in\Sigma$ arity one and
$\epsilon$ arity zero.

A term $t$ is \emph{linear} if no free variable appears more than once in $t$.
A linear term mentioning $k$ free variables is a \emph{$k$-ary context}.

\begin{definition}
Let $\FF$ be equipped with a partial order $\preceq$.  Then the lexicographic
partial order induced by $\preceq$ on $T(\FF)$ is the relation $\preceq$ defined as
follows: for any $f\in\FF_p,f'\in\FF_q$ and any $t_1,\ldots,t_p\in T(\FF)$ and
$t'_1,\ldots,t'_q\in T(\FF)$ we have $f(t_1,\ldots,t_p)\preceq
f'(t'_1,\ldots,t'_q)$ if and only if either $f \prec f'$ or $f=f'$ and $t_i\preceq
t'_i$ for all $i$.
\end{definition}

Note that this generalises Definition \ref{def:lex}, by taking $\epsilon\preceq a$
for all $a\in\Sigma$.  As before we will write $t \sim t'$ if $t, t'\in T(\FF)$ are related by the
lexicographic order; the definitions of \emph{chain} and \emph{antichain} are as
before.   To quantify antichain growth we need a notion of the size of a tree.
The measure we will use will be \emph{height}:

\begin{definition}The \emph{height} function $h:T(\FF,\XX)\rightarrow
\mathbb{N}$ is defined by $h(x) = 0$ for all $x\in\XX$, $h(t) = 1$ for all
$t\in\FF_0$ and $h(t(t_1,\ldots,t_n)) = 1 + \max(h(t_1,\ldots,t_n))$ for all
$t\in \FF_n$ ($n\geq 1$) and $t_1,\ldots,t_n\in T(\FF,\XX)$.  For a language
$L$, the set $\{t\in L\mid h(t)=k\}$ is denoted $L_{=k}$.
\end{definition}

For example, taking the earlier example of binary trees, ground terms of height
3 include $f(f(c,c),f(c,c))$, $f(c,f(c,c))$ and $g(f(c,c))$.

We say that $L$ has \emph{doubly exponential} antichain growth if there is some
$\epsilon$ such that the maximum size antichain in $L_{=n}$ exceeds
$2^{2^{\epsilon n}}$ infinitely often.

\begin{definition}
A \emph{nondeterministic finite tree automaton} (NFTA) over $\FF$ is a tuple
$\Aa=(Q,\FF,Q_f,\Delta)$ where $Q$ is a set of unary states, $Q_f\subseteq Q$ is
a set of final states, and $\Delta$ a set of transition rules of type 
\[f(q_1(x_1),\ldots,q_n(x_n)) \rightarrow q(f(x_1,\ldots,x_n)),\]
for $f\in\FF_n$, $q,q_1,\ldots,q_n\in Q$ and $x_1,\ldots,x_n\in\XX$.  The \emph{move
relation} $\arel$ is defined by applying a transition rule possibly inside a
context and possibly with substitutions for the $x_i$.  The reflexive transitive
closure of $\arel$ is denoted $\arels$.

A tree $t\in T(\FF)$ is \emph{accepted} by $\Aa$ if there is some $q\in Q_f$
such that $t\arels q(t)$.  The set of trees accepted by $\Aa$ is denoted
$\LL(\Aa)$.
\end{definition}

Again this generalises the definition of an NFA: put in transitions $\epsilon
\rightarrow q(\epsilon)$ for all accepting states $q$, $a(q(x))\rightarrow
q'(a(x))$ whenever $q\in \Delta(q',a)$, and set $Q_f$ as the initial state.

The critical idea for the proof is to find the appropriate analogue of $L_q$.
This turns out to be the set $P_q$ of binary contexts such that if the free
variables are assigned state $q$ then the root can also be given state $q$.
By analogy to the `trousers decomposition' of differential
geometry (also known as the `pants decomposition'), we refer to
such a context as a \emph{pair of trousers}.

It turns out that a sufficient condition for $L$ to have doubly exponential
antichain growth is for $P_q$ to be non-empty for some $q$ (note that this does
not depend on the particular partial order on $\Sigma$).  On the other hand, if
$P_q$ is empty for all $q$, then there is in a suitable sense no branching and
so we have a similar situation to ordinary languages.

\begin{definition}Let $\Aa = (Q,\FF,Q_f,\Delta)$ be an NFTA and $q\in Q$.  A
linear term $t\in T(\FF,\{x_1,x_2\})$ is a \emph{pair of trousers} with respect
to $q$ if $x_1,x_2$ appear in $t$ and $t[\sub{x_1}{q(x_1)}, \sub{x_2}{q(x_2)}]
\arels q(t)$.  The set of pairs of trousers with respect to
$q$ is denoted $P_q(\Aa)$.
\end{definition}

\begin{lemmarep}Let $\Aa = (Q,\FF,Q_f,\Delta)$ be a reduced NFTA.  If there exists some
$q\in Q$ such that $P_q(\Aa)$ is non-empty, then $\LL(\Aa)$ contains a doubly
exponential antichain.
\end{lemmarep}
\begin{proof}
We will clearly be done if we can find two pairs of trousers $t_1,t_2$ such that
$\sigma_1(t_1)\not\sim \sigma_2(t_2)$ for all substitutions
$\sigma_1,\sigma_2$: the set of trees built from them is of doubly exponential
size, and any two such trees are comparable only if they are constructed in
exactly the same way, i.e.  are equal.  We produce this pair by first
constructing two incomparable ground terms $s_1,s_2$ whose roots can be labelled
with state $q$.  Having done this we produce $t_1$ by attaching $s_1$ to the
left leg of our pair of trousers $t$, and a copy of $t$ to the right leg.  For
$t_2$ we do likewise but with $s_2$ in place of $s_1$.  Since $s_1\not\sim s_2$
we have that $\sigma_1(t_1)\not\sim \sigma_2(t_2)$ for all substitutions
$\sigma_1, \sigma_2$.

Let $t$ be a pair of trousers with respect to $q$ and let $s$ be a ground term
with $s\arels q(s)$.  We claim that there exist incomparable ground terms
$s_1,s_2$ with $s_i \arels q(s_i)$.

Indeed, we have that $s$ and $s' = t[\sub{x_1}{s},\sub{x_2}{s}]$ are ground
terms with $s\arels q(s)$ and $s'\arels q(s')$.  Let $s_1 =
t[\sub{x_1}{s},\sub{x_2}{s'}]$ and $s_2 = t[\sub{x_1}{s'},\sub{x_2}{s}]$.  Now
$s_1\preceq s_2$ only if $s\preceq s'$ and $s'\preceq s$, which is impossible as $s\neq
s'$ (since $h(s')>h(s)$).  Similarly we have that $s_2\not\preceq s_1$, as
required.

Hence $t_1=t[\sub{x_1}{s_1},\sub{x_2}{t}]$ and
$t_2=t[\sub{x_1}{s_2},\sub{x_2}{t}]$ are pairs of trousers with the property
that $\sigma_1(t_1)\not\sim \sigma_2(t_2)$ for all substitutions
$\sigma_1,\sigma_2$.  It is clear that a doubly exponential antichain can be
built from these.
\end{proof}

\begin{lemmarep}\label{lem:singexp}Let $\Aa = (Q,\FF,Q_f,\Delta)$ be a reduced NFTA such that
$P_q(\Aa)=\emptyset$ for all $q\in Q$.  Then $\LL(\Aa)$ has at most exponential
growth.
\end{lemmarep}
\begin{proof}
We proceed by induction on the number of states appearing on the left of
transitions.  Without loss of generality we may assume that $Q_f=\{q\}$ for some
$q$ (otherwise consider a finite union of automata).   Let $t\in \LL(\Aa)_{\leq
n}$ be any term of height at most $n$.  Say $t=f(t_1,\ldots,t_k)$ for some
function symbol $f$ and terms $t_1, \ldots, t_k$.  In any accepting run for $t$,
since the root is labelled with $q$ we have that $q$ can appear in at most one
subtree, since otherwise we obtain a pair of trousers.  Hence for all but at
most one value of $i$ we have that $t_i\in \LL(\Aa')$, where $\Aa'$ is $\Aa$
with all transitions in which $q$ appears on the left removed, which has at most
single exponential language growth by the inductive hypothesis.

Hence we have 
\[|\LL(\Aa)|_{\leq n}\leq |\FF|d|\LL(\Aa)|_{\leq n-1}
|\LL(\Aa')|_{\leq n-1}^d,\] 
where $d$ is the maximum arity of symbols in $\FF$.  Hence $\LL(\Aa)$ has at
most single exponential language growth.
\end{proof}

In the case where there are no pairs of trousers, the situation is essentially
equivalent to ordinary NFA, and so we have a further dichotomy between
exponential and polynomial antichain growth.  To show this, we define a set
equivalent to $L_{q,q}$, and show that we have polynomial growth if it is a
chain and exponential growth otherwise.

\begin{definition}Let $\Aa=(Q,\FF,Q_f,\Delta)$ be an NFTA, and $q\in Q$.  Define
$\LL_q(\Aa)\subseteq T(\FF,\{x_1\})$ to be the set of unary
contexts $t$ such that $t[\sub{x_1}{q(x_1)}]\arels q(t)$.
\end{definition}

Note that unary contexts are linear terms in which \emph{exactly} one free
variable appears, so $\LL_q(\A)$ does not contain ground terms.  Note also that
$x_1\in \LL_q(\A)$ for any $\A$.

To give meaning to the statement `$\LL_q(\A)$ is a chain', we must extend the
definition of the lexicographic order from the set $T(\FF)$ of ground terms to
the set $T(\FF,\{x_1\})$ of unary contexts.  We do this by extending the
relation $\preceq$ on $\FF$ to $\FF\cup\{x_1\}$ by $x_1\preceq f$ for all $f\in \FF$,
and extending this to the lexicographic order as before.

Note in particular we have that if $t=\sigma(t')$ for some substitution $\sigma$
then we have $t'\preceq t$; this corresponds to the notion of prefixes for words.
On the other hand, if $t'\preceq t$ then we have that either $t=\sigma(t')$ for
some $\sigma$ ($t'$ is a prefix of $t$) or otherwise that $\sigma'(t') \preceq
\sigma(t)$ for all substitutions $\sigma,\sigma'$.  Conversely, if $t\not
\sim t'$ then we have that $\sigma(t)\not\sim \sigma'(t')$ for all substitutions
$\sigma,\sigma'$; note that this does not hold for contexts of arity greater
than 1 (for a similar definition of the lexicographic order). 

\begin{lemmarep}\label{lem:polyexp}Let $\Aa=(Q,\FF,Q_f,\Delta)$ be a reduced NFTA
such that $P_q(\Aa)=\emptyset$ for all $q$.  Then $\LL(\Aa)$ has polynomial
antichain growth if $\LL_q(\Aa)$ is a chain for all $q$, and otherwise
$\LL(\Aa)$ has exponential antichain growth.
\end{lemmarep}
\begin{proof}
If $\LL_q(\A)$ is not a chain then let $t_1\not\sim t_2\in \LL_q(\A)$.  Since
$\A$ is reduced there is a ground term $t$ with $t\arels q(t)$ and a unary context $t'$
with $t'(q(x))\arels q'(t)$ for some $q'\in Q_f$.  Let the function
$\phi:\mathcal{P}(T(\FF))\rightarrow \mathcal{P}(T(\FF))$ be defined by
$\phi(X) = \left\{t_1[\sub{x_1}{s}],t_2[\sub{x_1}{s}] \middle| s\in X\right\}$,
and let $Y=\bigcup_{n=0}^{\infty} \phi^n(\{t\})$.  Then the set
$\left\{t'[\sub{x_1}{s}]\middle| s\in Y\right\} \subseteq \LL(\A)$ is an
antichain and has exponential growth.

Conversely if $\LL_q(\Aa)$ is a chain for all $q$ then an
argument similar to the upper bound in the proof of Theorem \ref{thm:dichotomy}
shows that $\LL(\Aa)$ has polynomial antichain growth.

Once again we proceed by induction on the number of states appearing on the left
of transitions, and assume without loss of generality that $Q_f=\{q\}$ for some
$q$.  Then for any $t\in \LL(\A)$ we have that $t=t'[\sub{x_1}{t''}]$ for some
$t'\in \LL_q(\A)$ and $t''\in \LL(\A')$, where $\A'$ is $\A$ with all
transitions in which $q$ appears on the left removed, which has polynomial
antichain growth by the inductive hypothesis.

For any antichain $L\subseteq \LL(\A)$, we claim that we have that \[L\subseteq
\left\{t[\sub{x_1}{t'}]\middle| t'\in \LL(\A')\right\}\] for some fixed $t\in
\LL_q(\A)$.  Indeed, supposing the contrary let $t_1 \neq t_2\in \LL_q(\A)$ be contexts
such that $t_1[\sub{x_1}{t'_1}],t_2[\sub{x_1}{t'_2}]\in L$ with $t_1\neq
\sigma(t_2), t_2\neq \sigma(t_1)$ for all substitutions $\sigma$.  Since
$\LL_q(\A)$ is a chain we have that (without loss of generality) $t_1\preceq t_2$
and since $t_1$ is not a prefix of $t_2$, we have that $\sigma_1(t_1)\preceq
\sigma_2(t_2)$ for all substitutions $\sigma_1,\sigma_2$.  In particular we have
that $t_1[\sub{x_1}{t'_1}] \preceq t_2[\sub{x_1}{t'_2}]$, which is a contradiction
since $L$ is an antichain, so the claim is proved.

Hence by induction we have that $\LL(\A)$ has polynomial antichain growth.
\end{proof}

Combining these lemmas gives 

\begin{theorem}\label{thm:treetrichot}
Let $L$ be a regular tree language over a partially ordered alphabet.  Then $L$
has either doubly exponential antichain growth, singly exponential antichain
growth, or polynomial antichain growth.
\end{theorem}

The special case of the trivial partial order (in which elements are only
comparable to themselves) yields the fact that the language growth of any
regular tree language is either polynomial, exponential or doubly exponential,
which may not have previously appeared in the literature.

\begin{corollary}
Let $L$ be a regular tree language.  Then $L$ has either doubly exponential
language growth, singly exponential language growth or polynomial language
growth.
\end{corollary}

Finally, we show that there is a polynomial algorithm to detect doubly
exponential growth, by determining whether or not the language of a given NFTA
contains a pair of trousers.

\begin{theoremrep}\label{thm:trousersalg}
There exists a polynomial time algorithm to determine whether the language of a
given NFTA has doubly exponential growth.
\end{theoremrep}
\begin{proof}
We show how to determine whether $P_{q_0}(\Aa)=\emptyset$ for fixed $q_0$.

We proceed similarly to the Reduction Algorithm in \cite{tata2007} (p.25), which
iteratively computes the set $M$ of states $q$ such that $t\arels q(t)$ for some
$t$.  We first iteratively compute the set $M'$ of states $q$ such that there is
a unary context $t\in T(\FF,\{x_1\})$ such that $t[\sub{x_1}{q_0}]\arels q$.
We can then iteratively compute the set $M''$ of states $q$ such that there is a
binary context $t\in T(\FF,\{x_1,x_2\})$ such that
$t[\sub{x_1}{q_0},\sub{x_2}{q_0}]\arels q$.  Then $T_{q_0}(\Aa)\neq \emptyset$
if and only if $q_0\in M''$.

Concretely, the reduction algorithm from \cite{tata2007} proceeds as follows.
Initialise the set $X=\emptyset$.  For each transition rule
$f(q_1(x_1),\ldots,q_n(x_n))\rightarrow q$ in $\Delta$ such that
$q_1,\ldots,q_n\in X$, add $q$ to $X$.  Repeat this process until $X$ no longer
changes.  Then $X=M$ is the set of accessible states.

To compute the set $M'$ of states $q$ such that $t[\sub{x_1}{q_0}]\arels q$ for
some unary context $t$, first initialise the set $X'=\{q_0\}$.  For each
transition rule $f(q_1(x_1),\ldots,q_n(x_n))\rightarrow q$ in $\Delta$ such that
we have $q_1,\ldots,q_{k-1},q_{k+1},\ldots,q_n \in M$ and $q_k\in X'$ for some
$k$, add $q$ to $X'$.  Repeat this until $X'$ no longer changes, and then we
have $X'=M'$.

Finally we compute the set $M''$ of states $q$ such that $t[\sub{x_1}{q_0},
\sub{x_2}{q_0}] \arels q$ for some binary context $t$.  First initialise $X''$
to be the set of states $q$ such that there is some transition rule $f(q_1(x_1),
\ldots,q_n(x_n)) \rightarrow q$ in $\Delta$ where $f$ has arity at least 2, and
we have $q_1,\ldots,q_{k-1},q_{k+1},\ldots,q_{l-1},q_{l+1},\ldots,q_n \in M$ and
$q_k,q_l\in M'$ for some $k<l$.

For the iterative step, for each transition rule $f(q_1(x_1),\ldots,q_n(x_n)) 
\rightarrow q$ in $\Delta$ such that we have $q_1,\ldots,q_{k-1},q_{k+1},\ldots,
q_n \in M$ and $q_k\in X''$ for some $k$, add $q$ to $X''$.  Repeat this until 
$X''$ stabilises and then we have $M''=X''$.  
\end{proof}

\section{Open problems}\label{sec:openprobs}

It is remarkable that, many decades after the discovery of the dichotomy between 
polynomial and exponential language growth, and 11 years after the work of 
Gawrychowski, Krieger, Rampersad and Shallit~\cite{gawrychowski2008finding}, it
remains unknown whether there is an efficient algorithm to compute the order 
of exponential language growth of a given NFA.  Consequently we consider that resolving 
this question (by providing either a polynomial-time algorithm or an appropriate
hardness result) is the most important open problem in this area.

For a DFA, on the other hand, the order of exponential language growth is easily 
computed as the spectral radius of the transition matrix.  However, it is not 
clear how such `algebraic' methods can be applied to the case of antichain 
growth, and so a second open problem is to find a polynomial-time algorithm to compute the 
order of exponential antichain growth for DFA.  Such a result would have 
immediate application to the field of quantified information flow, since it 
would allow one to compute the flow rate in the `dangerous' linear case, at the 
cost of determinising the automaton representing the system (with overhead 
roughly corresponding to the amount of hidden state the system contains).

The final problem in this direction is the combination of the preceding 
two: to find a polynomial-time algorithm to compute the order of exponential 
antichain growth for a given NFA. 

Alternatively we may wish to ask not about growth rates in the asymptotic limit, but instead about the 
precise width of $L_{=n}$ or $L_{\leq n}$ for given $n$.  This is particularly 
relevant to applications in computer security, where we may want not just an
approximation `for sufficiently large $n$' but a concrete guarantee.  For the case of a 
language given as a DFA and $n$ given in unary there is a straightforward 
dynamic programming algorithm to compute these quantities (for details see p.89 of 
\cite{mestelthesis}), but what about for NFA and for more concise representations of $n$?

Finally we pose a more speculative question: what other phenomena, apart from 
information flow, can antichains with respect to the lexicographic order 
usefully represent?

 \bibliographystyle{splncs04}
 \bibliography{mybib}
\end{document}